\documentclass[11pt]{article}
\usepackage{amsmath}
\usepackage{amssymb}
\usepackage{geometry}
\usepackage{amsthm}
\usepackage{amsfonts}
\usepackage{cite,url}
\usepackage{color}
\usepackage{autobreak}
\usepackage{vmargin}
\allowdisplaybreaks[1]

\newtheorem{Def}{Definition}[section]
\newtheorem{thm}{Theorem}[section]
\newtheorem{lem}{Lemma}[section]
\newtheorem{col}{Corollary}[section]
\newtheorem{rmk}{Remark}[section]
\newtheorem{Pro}{Proposition}[section]
\newcommand{\g}{\mathfrak{g}}

\numberwithin{equation}{section}

\title{On an infinite commuting ODE system associated \\ to a simple Lie algebra}
\date{\empty}
\author{Di Yang, \quad Cheng Zhang, \quad Zejun Zhou}
\begin{document}
\maketitle

	\begin{abstract}
Inspired by a recent work of Dubrovin~\cite{Du1}, for each simple Lie algebra $\mathfrak{g}$, we introduce an infinite family of pairwise commuting 
ODEs 
and define their $\tau$-functions. 
We show that these $\tau$-functions can be identified with the $\tau$-functions for the Drinfeld--Sokolov 
hierarchy of $\g$-type. 
Explicit examples for $\g=A_1$ and $A_2$ are provided, which are connected to the KdV hierarchy and the Boussinesq hierarchy respectively.
\end{abstract}



\tableofcontents
	
\section{Introduction}
For any affine Kac--Moody algebra, Drinfeld and Sokolov \cite{DS} introduced an integrable hierarchy of partial differential equations (PDEs) 
associated with a choice of vertex in the corresponding Dynkin diagram, now widely known as the 
{\it Drinfeld--Sokolov hierarchy}. 

Let $\g$ be a simple Lie algebra over $\mathbb C$. 
The Drinfeld--Sokolov hierarchy for the \mbox{untwisted} affine Kac--Moody algebra $\hat\g^{(1)}$ under a special  
choice of  vertex \cite{DS} (often denoted by $c_0$) is referred to as the {\it Drinfeld--Sokolov hierarchy of $\g$-type}  \cite{BDY2}, which 
is a generalization of the celebrated Korteweg--de Vries (KdV) hierarchy  (corresponding to $\g=\mathfrak{sl}_2(\mathbb{C})$).
The Drinfeld--Sokolov hierarchy of $\g$-type plays an important r\^ole in, for instance, the quantum singularity theory \cite{Du1, DLZ, FSZ, FJR, 
  LRZ, W1, W2}.    

In a recent study of the KdV hierarchy, Dubrovin~\cite{Du2} introduced an infinite dimensional ODE system,  which consists of 
infinitely many commuting hamiltonian equations, with the Poisson bracket given by an $r$-matrix structure.
He  defined $\tau$-functions for this   ODE system,
and established the relationship between this ODE system and the KdV hierarchy as well as the connection between their $\tau$-functions. An important consequence of the ODE system is to provide a direct way to approaching theta-functions as special $\tau$-functions of the KdV hierarchy \cite{Du2} (see also \cite{Du3}).   

The  content of this paper can be seen as a continuation of the program initiated by Dubrovin \cite{Du2}, 
namely, our goal is to give a generalization of Dubrovin's construction (corresponding to  $\g=\mathfrak{sl}_2(\mathbb{C})$)  to other simple Lie algebras. 
To be precise, we construct the ODE system for the Drinfeld--Sokolov hierarchy of~$\g$-type,  
define its $\tau$-functions and establish the relations of the $\tau$-functions to
those for the Drinfeld--Sokolov hierarchy of~$\g$-type.

Let a simple Lie algebra $\g$ be given. Denote by $n$ the rank and by $N$ the dimension. Let $h, h^{\vee}$ be the Coxeter and dual Coxeter numbers of $\g$, and let $m_1=1<m_2\leq\cdots\leq m_{n-1}<m_n=h-1$ be the exponents of $\g$~\cite{Kac}. 
Also denote by $(\cdot|\cdot)$ 
the normalized Cartan--Killing form, namely, 
\begin{equation}
		\left(x|y\right)=\frac{1}{2h^{\vee}}{\rm tr}\left(\rm ad_x\cdot ad_y\right),\quad x,y \in \mathfrak{g}\,. 
	\end{equation} 	 

Let   $\mathfrak{g}((\lambda^{-1}))$ be  the formal Laurent series in $\lambda^{-1}$ with coefficients in $\g$.  The normalized Cartan-Killing form  can be extended to   $\mathfrak{g}((\lambda^{-1}))$ through
\begin{equation}
	(x\lambda^{l_1}|y\lambda^{l_2})=(x|y)\lambda^{l_1+l_2}\,,\quad x,y\in \g\,,\quad l_1,l_2\in\mathbb{Z}\,.
\end{equation}
Recall from \cite{Kac, DS} that the {\em principal gradation} of $\mathfrak{g}((\lambda^{-1}))$ is defined by the following degree assignments:
\begin{equation}\label{Prin_degree}
\deg\lambda=h\,,\quad \deg E_i = -\deg F_i = 1\,, \quad i =1, \dots, n\,,
\end{equation} 
  where  $E_i$, $F_i$ are the Weyl generators of~$\g$.
Both $\g((\lambda^{-1}))$ and $\g$ can be decomposed into direct sums of subspaces with homogeneous degrees:
\begin{equation}
\g((\lambda^{-1})) = \widehat{\bigoplus_{k\in\mathbb{Z}}}\left(\g((\lambda^{-1}))\right)^{k}\,,\quad \g=\bigoplus_{k=-h+1}^{h-1}\g^{k}\,,
\end{equation}
where the upper index $k$ denotes the degree, namely, elements in $\g((\lambda^{-1}))^{k}$ and $\g^k$ are homegeneous of degree $k$.

Let $(e_j)_{j=1,\dots,N}$ be a  basis of~$\mathfrak{g}$, homogeneous with respect to the principal gradation. 
Denote $d_j=\deg e_j$ and $K_{ij}=\left(e_i|e_j\right)$,
and let $(e^i)_{i=1,\dots,N}$ be  the dual basis defined by $e^i=\sum_{j=1}^{N}K^{ij}e_j$ with $(K^{ij})=(K_{ij})^{-1}$. 
The structure constants $C^{ij}_k$ for $\g$ in terms of the dual basis are given by 
\begin{equation}
[e^i,e^j]=\sum_{k=1}^{N}C^{ij}_ke^k\,.
\end{equation}

Let $\Lambda(\lambda)$ denote the {\em cyclic element} of   $\mathfrak{g}((\lambda^{-1}))$  
 \cite{DS, Kac1978,Kostant} (which will be defined in \eqref{LambdaDef} in Section \ref{sec:3}). Consider a $\g$-valued formal Laurent series $W(\lambda)$ as follows:
\begin{equation}\label{WDef}
W(\lambda)=\Lambda(\lambda)+\sum_{j=1}^N \sum _{   k\in {\mathbb Z},\, k\geq \frac{d_j}{h}  -1  }  u_k^j \,e_j \, \lambda^{-k-1}\,, 
	\end{equation} 
        where  $u_k^j$ are indeterminates, and let $\mathcal{P}$ be the polynomial ring in $u_k^j$ over $\mathbb{C}$:
        \begin{equation}
\mathcal{P}=\mathbb{C}\left[u^j_k\,\Big|\,j=1,\dots,N,\, k\in \mathbb{Z}, \,k \geq \frac{d_j}{h}-1\right]\,.
        \end{equation}
   
Based on the $r$-matrix structure for     $\mathfrak{g}((\lambda^{-1}))$  given in \cite{FT-book}, a Poisson bracket on~$\mathcal{P}$ can be defined as follows: 
\begin{equation}\label{r-matixPoisson}
		\left\{W(\lambda) \underset{'}{\otimes} W(\mu)\right\}=\left[r(\lambda-\mu),W(\lambda)\otimes I+I\otimes W(\mu)\right],
\end{equation} 
where $I$ denotes the  identity matrix and
\begin{equation}\label{rmatrixDef}
r(\lambda)=\sum_{j=1}^{N}\frac{e^j\otimes e_j}{\lambda}\,.  
\end{equation} 
Here $r(\lambda) $ is the well-known rational $r$-matrix which is a solution to the classical Yang--Baxter equation  \cite{BD-r-matrix, Sklyanin, FT-book}.
Alternatively, \eqref{r-matixPoisson} can be written as
\begin{equation}\label{PoiBra2}
	\left\{u^i(\lambda),u^j(\mu)\right\}=\sum_{k=1}^{N}C^{ij}_{k}\frac{u^k(\lambda)-u^k(\mu)}{\lambda-\mu}+\frac{\left([e^i,e^j]|\Lambda(\lambda)-\Lambda(\mu)\right)}{\lambda-\mu}\,,
      \end{equation}
  where
\begin{equation}\label{eq:ujl}
		u^j(\lambda):=\sum_{k\in\mathbb{Z},\, k\geq\frac{d_j}{h}-1}u^j_k\lambda^{-k-1}\,.
\end{equation}
By direct computations,  
 one can check that the bracket \eqref{r-matixPoisson} (or, equivalently \eqref{PoiBra2}) is well defined.  In Section \ref{sec:3}, we will also provide a direct verification that this bracket is indeed a Poisson bracket (see Proposition~\ref{PiossonBracketSEC2} and Remark \ref{rmk:21}).   

The following lemma, which  is an analogue of  \cite[Proposition 4.1]{DS}, is of importance. 
	\begin{lem}\label{FuLem}
          There exists a unique pair $\left(U(\lambda),H(\lambda)\right)\in \left(\mathcal{P}\otimes \mathfrak{g}((\lambda^{-1}))\right)^2$ of the form
          \begin{align}\label{UHform-u}
            U(\lambda)=\sum_{k\geq1}U^{[-k]}(\lambda), &\quad U^{[-k]}(\lambda)\in\mathcal{P}\otimes{\left(\rm Im\,ad_{\Lambda(\lambda)}\right)}^{-k}\,,\\
\label{UHform-h}            H(\lambda)=\sum_{k\geq0}H^{[-k]}(\lambda),& \quad H^{[-k]}(\lambda)\in \mathcal{P}\otimes{\left(\rm Ker\,ad_{\Lambda(\lambda)}\right)}^{-k}\,,
          \end{align}
          such that
          \begin{equation}\label{Flem}
            e^{-{\rm ad}_{U(\lambda)}}W(\lambda)=\Lambda(\lambda)+H(\lambda)\,.
          \end{equation} 
	\end{lem}
	
        The proof will be  given in Section \ref{sec:3}.

        It is known from \cite{Kac} that ${\rm Ker\, ad}_{\Lambda(\lambda)}$ is an infinite dimensional abelian Lie subalgebra of $ \mathfrak{g}((\lambda^{-1}))$, and there exist $\Lambda_j(\lambda)\neq 0$, $j\in E$, such that
\begin{equation}\label{KeradDecomposition}
	{\rm Ker\, ad}_{\Lambda(\lambda)}=\widehat{\bigoplus_{j\in E}}\mathbb{C}\Lambda_j(\lambda)\,, \quad \deg\Lambda_j(\lambda)=j\,,
\end{equation}
where  $E=\sqcup_{a=1}^{n}(m_a+h\mathbb{Z})$ (we remind the reader that  $m_1, \dots, m_n$ are the exponents of~$\g$). The normalization of $\Lambda_j $, $j\in E$, is chosen as follows: 
\begin{equation}\label{KerNormalLambda_a}
	\Lambda_{m_a+kh}(\lambda)=\lambda^k\Lambda_{m_a}(\lambda),\quad \Lambda_1(\lambda)=\Lambda(\lambda),\quad  \left(\Lambda_{m_a}(\lambda)|\Lambda_{m_b}(\lambda)\right)=h\eta_{ab}\lambda\,.
\end{equation}	
Here and below $\eta_{ab}:=\delta_{a+b,n+1}$.
By \eqref{KeradDecomposition}, we can write 
	\begin{equation}\label{H_hLambda_j}
	  H(\lambda)=:\sum_{a=1}^{n}\sum_{k\geq-1 }\frac{h_{n+1-a,k}}{h}\Lambda_{m_a-kh-2h}(\lambda)\,,\quad h_{a,k}\in\mathcal{P}\,.
\end{equation}

 \begin{Def}\label{def:1}
Define 
an infinite hamiltonian ODE system as follows:
\begin{equation}\label{ODE intro}
	\frac{d W(\lambda)}{d T^a_k}= \bigl\{h_{a,k},W(\lambda)\bigr\}, \quad a=1,\dots,n,\,k\geq 0\,,
      \end{equation}
where $h_{a,k}$, given in \eqref{H_hLambda_j},  are the hamiltonians and $T^a_k$ are the associated
time variables.
\end{Def}  

Define 
$R_a(\lambda)\in\mathcal{P}\otimes\mathfrak{g}((\lambda^{-1})),a=1,\dots, n$, by 
\begin{equation}\label{DefBasicResolvent}
R_a(\lambda):=e^{{\rm ad}_{U(\lambda)}}\Lambda_{m_a}(\lambda)\,.
\end{equation}
Analogous to~\cite{BDY2, Du2, DS}, we call $R_a(\lambda)$ the {\it basic $\mathfrak{g}$-resolvents of $W(\lambda)$}.

\begin{thm}\label{thm:1} 
The 
 ODE system~\eqref{ODE intro} can be written as
\begin{equation}\label{eq:ODEhamil}
\frac{dW(\lambda)}{dT^{a}_k}=\left[\left(\lambda^k R_a(\lambda)\right)_+,W(\lambda)\right], \quad a=1,\dots,n,\,k\geq 0\,,
\end{equation}
where $(\,)_+$ means the polynomial part in~$\lambda$.
\end{thm}

\begin{thm}\label{thm:2}
  The hamiltonians are in involution, i.e., 
 \begin{equation}
	\left\{h_{a,k},h_{b,l}\right\}=0,  \quad \forall\,a,b=1,\dots,n, \, k,l\geq 0\,.
\end{equation}
\end{thm}
The proofs of the above theorems will be given in Section \ref{sec:3}.

 Following~\cite{BDY2, Du2}, define $\omega_{a,k;b,l}\in\mathcal{P}$ by means of generating series as follows:
\begin{equation}\label{ODE_tauStucture}
\sum_{k,l\geq 0}\frac{\omega_{a,k;b,l}}{\lambda^{k+1}\mu^{l+1}}=\frac{\left(R_a(\lambda)|R_b(\mu)\right)}{\left(\lambda-\mu\right)^2}-\eta_{ab}\frac{m_a\lambda+m_b\mu}{\left(\lambda-\mu\right)^2} \,, \quad a,b=1,\dots,n\,.
\end{equation}
We will show in Lemma \ref{lemmaOmegaSys}  that the polynomials $\omega_{a,k;b,l}$ satisfy the following identities:
\begin{equation}\label{OmegaSys}
		\omega_{a,k;b,l}=\omega_{b,l;a,k},\quad\frac{d\omega_{a,k;b,l}}{d T^c_m}=\frac{d\omega_{c,m;a,k}}{d T^b_l}=\frac{d\omega_{b,l;c,m}}{d T^a_k}\,,
	\end{equation} 
for $a,b,c=1,\dots,n$,  $k$, $l$, $m\geq 0$. We call $\omega_{a,k;b,l}$ the \textit{$\tau$-structure} for the 
 ODE system~\eqref{ODE intro}.

Due to \eqref{OmegaSys}, for any  solution $W(\lambda,\mathbf{T})$ to \eqref{ODE intro}, there exists a function $\tau:=\tau(\mathbf{T})$ satisfying
\begin{equation}\label{deftau}
\frac{d^2 \log \tau(\mathbf{T})}{dT^a_kdT^b_l}=\omega_{a,k;b,l}(\mathbf{T})\,,\quad a,b=1,\dots,n, \, k,l\geq 0\,,
\end{equation}
where $\mathbf{T}$ denotes the set of time variables $T^a_k$, $a=1,\dots,n $, $k\geq 0$. 
We call $\tau(\mathbf{T})$ the {\it $\tau$-function of the solution $W(\lambda,\mathbf{T})$} to the ODE system~\eqref{ODE intro},
which is determined by $W(\lambda,\mathbf{T})$ up to a multiplicative factor of the form \begin{equation}\label{tauFactor}
	 e^{d_0+\sum_{a=1}^{n}\sum_{k\geq 0}d_{a,k}T^a_k}\,,
\end{equation}where $d_0,d_{a,k}$ are arbitrary constants. 

The following theorem  will be proved in Section \ref{section4}. 
\begin{thm}\label{thm:3}
Every $\tau$-function $\tau(\mathbf{T})$ for the ODE system~\eqref{ODE intro} is a $\tau$-function for the Drinfeld--Sokolov hierarchy of $\g$-type; vice versa. 
\end{thm}

The rest of the paper is organized as follows. In Section~\ref{sec:3}, we provide details of the construction of the  ODE system~\eqref{ODE intro}, 
and prove Lemma~\ref{FuLem}, Theorem~\ref{thm:1} and Theorem~\ref{thm:2}. 
In Section~\ref{sec:4}, we study the $\tau$-structure for the ODE system~\eqref{ODE intro}. 
Section \ref{section4} aims to prove Theorem~\ref{thm:3}. In Section \ref{sec:2},
explicit examples for $\g = A_1$ and $A_2$ are provided. 
Finally, in Section~\ref{sec:concl} are gathered our conclusions which are put into perspectives.

\section{An infinite commuting ODE system associated to $\g$} 
\label{sec:3}	
This section aims to give  details of the construction of the  ODE system~\eqref{ODE intro}. 

Let a simple Lie algebra $\g$ be given. Fix a Cartan subalgebra $\mathfrak{h}$ of $\g$, and let $\Delta \subset \mathfrak{h}^*$ be the root system. We choose a set of simple roots $\Pi=\left\{\alpha_1,\dots,\alpha_n\right\}\subset \mathfrak{h}^*$. Then, $\g$ has the root space decomposition:
 \begin{equation}
		\mathfrak{g}=\mathfrak{h}\oplus_{\alpha\in\triangle}\mathfrak{g}_\alpha\,.
	\end{equation}
For any $\alpha\in \Delta$, denote by $H_\alpha$ the unique element in $\g$ such that $\left(H_\alpha|X\right)=\alpha(X),\forall X\in\g$. The normalized Cartan--Killing form induces a non-degenerate bilinear form on $\mathfrak{h}^*$ through $\left(\alpha|\beta\right)=\left(H_\alpha|H_\beta\right),\alpha,\beta\in \Delta$. 

Let $\theta$ be the highest root with respect to $\Pi$. We can choose $E_\theta\in \g_{\theta}$, $E_{-\theta}\in \g_{-\theta}$, normalized by the conditions $(E_{\theta}|E_{-\theta})=1$ and $\omega(E_{-\theta})=-E_{\theta}$, where $\omega : \g\rightarrow\g$ is the Chevalley involution. Let 
\begin{equation}
	I_+:=\sum_{a=1}^{n}E_a
\end{equation} be the principal nilpotent element of $\g$, where $E_1, \dots, E_n$ are  the Weyl generators ({\it cf.}~\eqref{Prin_degree}).  The cyclic element  $\Lambda(\lambda)$ of   $\mathfrak{g}((\lambda^{-1}))$ is defined as \cite{DS, Kac1978,Kostant}
\begin{equation}\label{LambdaDef}
\Lambda(\lambda)=I_++\lambda E_{-\theta}\,.
\end{equation}

\begin{Pro}\label{PiossonBracketSEC2}
	The bracket \eqref{r-matixPoisson} (or equivalently \eqref{PoiBra2}) is a Poisson bracket.
\end{Pro}
\begin{proof}
The bilinearity and skew-symmetry follows directly from \eqref{PoiBra2}. Computing the quantity $\{u^i(\lambda),\{u^j(\mu),u^k(\nu)\}\}$, one has
\[
  \{u^i(\lambda),\{u^j(\mu),u^k(\nu)\}\}=\sum_{l=1}^{N}\sum_{s=1}^{N}C^{jk}_lC^{il}_s\,\frac{u^s(\lambda)(\mu-\nu)+u^s(\mu)(\nu-\lambda)+u^s(\nu)(\lambda-\mu)}{(\lambda-\mu)(\lambda-\nu)(\mu-\nu)}\,.\]
Then, the Jacobi identity follows from that of $\g$.
The proposition is proved.
\end{proof}

\begin{rmk}\label{rmk:21}
Note that similar formula to  \eqref{r-matixPoisson} given in  \cite[Part~Two, Chapter~IV]{FT-book} defines a Poisson structure for    $\mathfrak{g}((\lambda^{-1}))$. Then, Proposition \ref{PiossonBracketSEC2} holds because of the well-definedness of \eqref{r-matixPoisson} for  
our specific form of $W(\lambda)$ given in \eqref{WDef}. 
\end{rmk}

Now we proceed to the proof of Lemma~\ref{FuLem}.

\begin{proof}[Proof of Lemma~\ref{FuLem}] 
Comparing the components with principal degree  $-k$ on both sides of \eqref{Flem}, we find that 
\eqref{Flem} is equivalent to 
\begin{equation}\label{adWdegk}
 H^{[-k]}(\lambda)
=\sum_{j=-1}^{k}\sum_{s=0}^{k-j} (-1)^s \sum_{1\leq l_1,\dots,l_s\leq k-j\atop l_1+\dots+l_s=k-j}\frac{{\rm ad}_{U^{[{-l_1}]}(\lambda)}\cdots {\rm ad}_{U^{[-l_s]}(\lambda)}}{s!}W^{[-j]}(\lambda)\,,\quad k\ge0.
\end{equation}
Here, $W^{[1]}(\lambda):=\Lambda(\lambda)$. The rest of the proof given by an induction on~$k$ 
	is exactly similar to the one for the analogous lemma~\cite{BDY2} (see also~\cite{DS}).
\end{proof}

	 For convenience, we introduce the generating series of hamiltonians 
	 \begin{equation}\label{hamiltonianSeries}
	 	h_a(\lambda):=h\lambda \delta_{a,n}+\sum_{k\geq -1}h_{a,k}\lambda^{-k-1}\,. 
	 \end{equation}
	 Using \eqref{KerNormalLambda_a} and \eqref{H_hLambda_j}, one gets 
	 \begin{equation}\label{HhakLambda}
\Lambda(\lambda)+H(\lambda)=
                \sum_{a=1}^{n}\frac{h_{n+1-a}(\lambda)}{h\lambda}\Lambda_{m_a}(\lambda)\,.
	\end{equation}
 Lemma \ref{FuLem} together with $R_{a}(\lambda)$, which is the basic $\g$-resolvent defined in \eqref{DefBasicResolvent}, allows us to express  $W(\lambda)$  in the form
\begin{equation}\label{Sec3WhR}
		W(\lambda)=\sum_{a=1}^{n}\frac{h_{n+1-a}(\lambda)}{h\lambda}R_{a}(\lambda)\,. 
              \end{equation}

              Then, we can prove by  straightforward computations the following lemma, which will be useful in later sections.
\begin{lem}\label{lemRR}
		Let $a,b=1,\dots,n$, we have
 		\begin{align}
			(R_a(\lambda)|R_b(\lambda))&=\eta_{ab}h\lambda\,,\label{TrRR}\\
			(W(\lambda)|R_a(\lambda))&=h_a(\lambda)\,,\label{TrWR}
			\\\left[R_a(\lambda),R_b(\lambda)\right]&=0\,,\label{RRcommutator}
			\\ \left[W(\lambda), R_a(\lambda)\right]&=0\,.\label{WRcommutator}
		\end{align}
 	\end{lem}

        Based on Lemma \ref{FuLem}, we are in the position to prove Theorem \ref{thm:1}.
\begin{proof}[Proof of Theorem \ref{thm:1}]
Using \eqref{hamiltonianSeries} and Definition \ref{def:1}, to show \eqref{eq:ODEhamil} it suffices to show  
  \begin{equation}\label{eq:hbW2}
    		\{h_b(\mu),W(\lambda)\} = \frac{[R_b(\mu),W(\lambda)]}{\mu-\lambda}\,. 
  \end{equation}
 Using \eqref{WDef},  \eqref{KerNormalLambda_a} and  \eqref{HhakLambda}, one has
  \begin{equation}
		\{h_b(\mu),W(\lambda)\}=
		\sum_{j=1}^{N} \{h_b(\mu),u^j(\lambda)\}e_j 
      =      	\sum_{j=1}^{N}                    (\{H(\mu),u^j(\lambda)\}|\Lambda_{m_b}(\mu))e_j\,.  \end{equation}
By noticing that  \begin{align}\label{eq:HU}
		\{H(\mu),u^j(\lambda)\}&=\{e^{-{\rm ad}_{U(\mu)}}W(\mu),u^j(\lambda)\} \nonumber \\ &=[S_{j}(\mu,\lambda),\Lambda(\mu)+H(\mu)]-e^{-{\rm ad}_{U(\mu)}}\{u^j(\lambda),W(\mu)\}\,,
              \end{align}
        and by the ad-invariance of the Cartan-Killing form and \eqref{DefBasicResolvent}, we find that
        \begin{equation}\label{eq:hbw1}
          \{h_b(\mu),W(\lambda)\}=		-\sum_{j=1}^{N}  (\{u^j(\lambda),W(\mu)\}|R_b(\mu)) e_j\,.
        \end{equation}
Here we have used Lemma \ref{FuLem} in the derivation of  \eqref{eq:HU}, and        $S_{j}(\mu,\lambda)$ is given by
        \begin{equation}
        S_{j}(\mu,\lambda):=\sum_{l\ge 0}\frac{(-{\rm ad}_{U(\mu)})^l}{(l+1)!}\{u^j(\lambda),U(\mu)\}          \,. 
        \end{equation}
	From \eqref{PoiBra2}, one can obtain
	 \begin{equation}\label{eq:ujW}
	\{u^j(\lambda),W(\mu)\}=-\frac{[e^j,W(\lambda)-W(\mu)]}{\lambda-\mu}\,,\quad j=1,\dots,N\,.
\end{equation}
Inserting \eqref{eq:ujW} into \eqref{eq:hbw1} leads to \eqref{eq:hbW2}, which completes the proof. \end{proof}

Let us now prove Theorem \ref{thm:2}.
\begin{proof}[Proof of Theorem~\ref{thm:2}]
	Applying $e^{-{\rm ad}_{U(\lambda)}}$ to \eqref{eq:ODEhamil} and using \eqref{ODE intro}, one obtains
	\begin{equation}\label{proofthm2}
		\{h_{a,k},H(\lambda)\}+\left[\sigma_{a,k}(\lambda),\Lambda(\lambda)+H(\lambda)\right]=0\,,\end{equation} 
	where
	\begin{equation}\sigma_{a,k}(\lambda):=\sum_{l\geq 0}\frac{\left(-{\rm ad}_{U(\lambda)}\right)^l}{(l+1)!}\{h_{a,k}, U(\lambda)\}-e^{-{\rm ad}_{U(\lambda)}}(\lambda^kR_a(\lambda))_+\,.
	\end{equation}
Then, one gets
\begin{equation}
	\left\{h_{a,k},h_b(\lambda)\right\}=\left(\{h_{a,k},H(\lambda)\}|\Lambda_{m_b}(\lambda)\right)=-\left(\left[\sigma_{a,k}(\lambda),\Lambda(\lambda)+H(\lambda)\right]|\Lambda_{m_b}(\lambda)\right)=0\,, 
\end{equation} 
where the first equality follows from \eqref{HhakLambda}, and the last equality relies on the ad-invariance of the Cartan-Killing form.  The theorem is proved.
\end{proof}


 Define the loop operator \begin{equation}
\nabla_a(\lambda):=\sum_{k\geq0}\frac{1}{\lambda^{k+1}}\frac{d}{d T^a_k}\,,\quad a=1,\dots,n\,.
\end{equation}
The following lemma, which is analogous to  \cite[Lemma 2.2.7]{BDY2},  will be useful in Section \ref{sec:4}. 
 \begin{lem}\label{LoopRa}We have
\begin{equation}\label{eq:loopR}
	\nabla_a(\lambda)R_b(\mu)=\frac{\left[R_a(\lambda),R_b(\mu)\right]}{\lambda-\mu}-\left[Q_a(\lambda),R_b(\mu)\right],\quad a,b=1,\dots n\,,
\end{equation}
where $Q_a(\lambda):=\left(\lambda^{-1}R_a(\lambda)\right)_+$.
\end{lem}
\begin{proof}
  Using \eqref{proofthm2} and Theorem \ref{thm:2}, we have
  \begin{equation}\label{S_2LH}
    \left[\sigma_{a,k}(\lambda),\Lambda(\lambda)+H(\lambda)\right]=0\,,\quad a=1,\dots,n,\, k\geq0 \,.
  \end{equation}
Using an argument similar to the one in \cite{BDY2}, it follows from \eqref{S_2LH} that  $\sigma_{a,k}(\lambda)\in\mathcal{P}\otimes{\rm Ker}\,{\rm ad}_{\Lambda(\lambda)}$.  
Then, 
  \begin{equation}
	\left[\sigma_{a,k}(\lambda),\Lambda_{m_b}(\lambda)\right]+\{h_{a,k},\Lambda_{m_b}(\lambda)\}=0\,,\quad  a,b=1,\dots,n,\,k\geq0\,. 
\end{equation}
Applying $e^{\text{ad}_{U(\lambda)}}$ to both sides of the above equation, one obtains
	\begin{equation}\label{RRZoreCur}
		\left\{h_{a,k},R_b(\lambda)\right\}=\left[\left(\lambda^kR_a(\lambda)\right)_+,R_b(\lambda)\right]\,.
	\end{equation}
 This is equivalent to \eqref{eq:loopR}.
\end{proof}

\begin{rmk}
	We note that $Q_a(\lambda)\in\mathcal{P}\otimes\g$, namely, $ Q_a(\lambda)$ is independent of $\lambda$.
\end{rmk}

\section{$\tau$-structure of the ODE system~\eqref{ODE intro}}
\label{sec:4}
In this section, we provide the detailed construction of $\tau$-functions for the  ODE system~\eqref{ODE intro}, along 
the line similar to that for the Drinfeld--Sokolov hierarchy~\cite{BDY2}.

Let us start with proving that $\omega_{a,k;b,l}$ are indeed well defined through \eqref{ODE_tauStucture}. 
Substituting the following equality
\begin{equation}
R_b(\mu)=R_b(\lambda)+R_b'(\lambda)(\mu-\lambda)+(\mu-\lambda)^2\partial_\lambda\biggl(\frac{R_b(\lambda)-R_b(\mu)}{\lambda-\mu}\biggr)
\end{equation}
into the right-hand side of \eqref{ODE_tauStucture}, we obtain
 \begin{align}
   &\frac{(R_a(\lambda)|R_b(\mu))}{(\lambda-\mu)^2} -\eta_{ab}\frac{m_a\lambda+m_b\mu}{(\lambda-\mu)^2} \nonumber\\
   =&\frac{\eta_{ab}h\lambda}{(\lambda-\mu)^2}-\frac{(R_a(\lambda)|R'_b(\lambda))}{\lambda-\mu}+\biggl(R_a(\lambda) \Big|\partial_\lambda\biggl(\frac{R_b(\lambda)-R_b(\mu)}{\lambda-\mu}\biggr)\biggr)-\eta_{ab}\frac{m_a\lambda+m_b\mu}{(\lambda-\mu)^2}\,. \label{rexpandsub}
\end{align}
Introduce 
$$F(t,\lambda)=\Bigl(e^{t{\rm ad}_{U(\lambda)}}\Lambda_{m_a}(\lambda)\Big| \bigl(e^{t{\rm ad}_{U(\lambda)}}\Lambda_{m_b}(\lambda)\bigr)' \Bigr)\,,$$
where prime, ``$\,'\,$", denotes the derivative with respect to $\lambda$. By a direct computation, one 
can prove that $\frac{d F(t,\lambda)}{dt}\equiv 0$. 
Therefore, 
\begin{equation}\label{rarbprime}
(R_a(\lambda)|R_b'(\lambda))=F(1,\lambda)=F(0,\lambda)=\left(\Lambda_{m_a}(\lambda)|\Lambda_{m_b}'(\lambda)\right)=\eta_{ab}m_b\,.
\end{equation}
Formulae \eqref{rarbprime} and~\eqref{rexpandsub} lead to the well-definedness ({\it cf}.~\cite{BDY2}).

We have the following lemma, as claimed in Introduction. 
\begin{lem}\label{lemmaOmegaSys}
  The polynomials $\omega_{a,k;b,l}$ satisfy the identity \eqref{OmegaSys}.
\end{lem}
\begin{proof}
Using Lemma \ref{LoopRa} and the definition \eqref{ODE_tauStucture}, one can deduce 
 \eqref{OmegaSys}.
\end{proof}
\begin{rmk}
 It can be proved that 
	\begin{equation}
		\left\{h_{a,-1},R_b(\lambda)\right\}=\left[Q_a(\lambda),R_b(\lambda)\right]\,,\quad a=1,\dots,n\,.
\end{equation}
Then by a direct computation one can show that 
\begin{equation}
\{h_{c,-1}, \, \omega_{a,k;b,l}\}
=0\,, \quad a,b,c=1,\dots,n,\,k,l\geq 0\,.
\end{equation}
Thus, the hamiltonian flows given by $h_{a,-1}$, $a=1,\dots,n$,  are not included in the 
definition of the $\tau$-structure.
\end{rmk}

Using the $\tau$-structure $\omega_{a,k;b,l}$, for any solution $W(\lambda,\mathbf{T})$ to \eqref{ODE intro}, 
we can define the $\tau$-function of the solution $W(\lambda,\mathbf{T})$ using~\eqref{deftau} 
as given in Introduction. 

For $m\geq3$, define 
\begin{equation}
\omega_{a_1,k_1;\dots;a_m,k_m}:=\frac{d^m\log\tau}{dT^{a_1}_{k_1}\dots d T^{a_m}_{k_m}}\,,\quad a_1,\dots,a_m=1,\dots,n,\,k_1,\dots,k_m\geq 0\,,
\end{equation}
and denote by $B$ the following multi-linear form: 
\begin{equation}
  \label{mLinearFormB}
B(x_1,\dots,x_m):={\rm tr}({\rm ad}_{x_1}\dots {\rm ad}_{x_m})\,, \quad x_1,\dots,x_m\in\g\,.
\end{equation}

Using Lemma \ref{LoopRa}, we can prove the following proposition.
 \begin{Pro}\label{Ncorelater}
For $m\geq2$, the following identity holds true: 
\begin{align}
		\sum_{k_1,\dots,k_m\geq 0}\frac{\omega_{a_1,k_1;\dots;a_m,k_m}}{\lambda_1^{k_1+1}\cdots\lambda_m^{k_m+1}}=-\frac{1}{2mh^{\vee}}\sum_{s\in S_m}\frac{B\bigl(R_{a_{s_1}}(\lambda_{s_1}),\dots,R_{a_{s_m}}(\lambda_{s_m})\bigr)}{\prod_{i=1}^{m}(\lambda_{s_i}-\lambda_{s_{i+1}})}\nonumber\\
		-\delta_{m,2}\eta_{a_1a_2}\frac{m_{a_1}\lambda_1+m_{a_2}\lambda_2}{(\lambda_1-\lambda_2)^2}\,,\label{mPointCorrelator}
\end{align}
where $s_{m+1}$ is understood as $s_1$. 
\end{Pro}

For an arbitrary solution $W(\lambda,\mathbf{T})$ to the  ODE system~\eqref{ODE intro}, both sides of the identity~\eqref{mPointCorrelator} can  be regarded as elements in 
	$\mathbb{C}[[\mathbf{T}]][[\lambda_1^{-1},\dots,\lambda_m^{-1}]]$. 
	Then by taking $\mathbf{T}=\mathbf{0}$, we obtain the following corollary.
\begin{col}\label{TaylorTauODE}
For $m\geq 2$, the generating series of $m$-th order Taylor coefficients of $\log\tau({\bf T})$ has the explicit expression: 
	\begin{align}
			\sum_{k_1,\dots,k_m\geq 0}\frac{\frac{d^m\log\tau}{dT^{a_1}_{k_1}\dots d T^{a_m}_{k_m}}(\mathbf{0})}{\lambda_1^{k_1+1}\cdots\lambda_m^{k_m+1}}=-\frac{1}{2mh^{\vee}}\sum_{s\in S_m}\frac{B(R_{a_{s_1}}(\lambda_{s_1},\mathbf{0}),\dots,R_{a_{s_m}}(\lambda_{s_m},\mathbf{0}))}{\prod_{i=1}^{m}(\lambda_{s_i}-\lambda_{s_{i+1}})} \nonumber\\
			-\delta_{m,2}\eta_{a_1a_2}\frac{m_{a_1}\lambda_1+m_{a_2}\lambda_2}{(\lambda_1-\lambda_2)^2}\,.\label{mtauPointCorrelator}
	\end{align}
\end{col}

\section{Proof of Theorem~\ref{thm:3}}\label{section4}
\subsection{Review of the Drinfeld--Sokolov hierarchy}
In this subsection, let us recall the construction of a hierarchy of integrable PDEs given by Drinfeld and Sokolov \cite{DS}.

Denote by $\mathfrak{b}=\g^{\leq0}$ the Borel subalgebra of $\g$, and $\mathfrak{n}=\g^{<0}$ the nilpotent subalgebra.  The Lax operator of {\em pre-Drinfeld--Sokolov hierarchy} introduced in \cite{DS} is as follows:
\begin{equation}\label{LaxOP}
	\mathcal{L}(\lambda)=\partial_x+\Lambda(\lambda)+q(x)\,,\quad q(x)\in\mathfrak{b}\,, 
      \end{equation}
      where $\Lambda(\lambda)$ is the cyclic element as given in Section \ref{sec:3}. 

 Let $\mathcal{A}^q$ be the ring of differential polynomials in $q$. Namely, elements of $\mathcal{A}^q$ are polynomials of  the entries of $q,q_x,q_{xx},\dots$. 
It was proved  in \cite{DS}  that there exsits a unique pair ($\mathcal{U}(\lambda),\mathcal{H}(\lambda))\in (\mathcal{A}^q\otimes \mathfrak{g}((\lambda^{-1})))^2$ of the form:
\begin{align}\label{UHformPDE}
		\mathcal{U}(\lambda)&=\sum_{k\geq 1}\mathcal{U}^{[-k]}(\lambda),\quad \mathcal{U}^{[-k]}(\lambda)\in \mathcal{A}^q\otimes{\left(\rm Im\,ad_{\Lambda(\lambda)}\right)}^{-k}\,, \\
		\mathcal{H}(\lambda)&=\sum_{k\geq0}\mathcal{H}^{[-k]}(\lambda),\quad \mathcal{H}^{[-k]}(\lambda)\in \mathcal{A}^q\otimes{\left(\rm Ker\,ad_{\Lambda(\lambda)}\right)}^{-k}\,,
\end{align}
such that
\begin{equation}\label{PDEflemmeq}
	e^{-{\rm ad}_\mathcal{U(\lambda)}}\mathcal{L}(\lambda)=\partial_x+\Lambda(\lambda)+\mathcal{H}(\lambda)\,.
\end{equation}

Let 
 \begin{equation}\label{PDEresolvent}
R^{PDE}_a(\lambda):=e^{{\rm ad}_\mathcal{U(\lambda)}}\Lambda_{m_a}(\lambda)\,, \quad a=1,\dots,n\,.
\end{equation}
Recall that the pre-Drinfeld--Sokolov hierarchy \cite{DS} is a  commuting family of PDEs for the $\mathfrak{b}$-valued function $q=q(x,\mathbf{T})$, $\mathbf{T}=(T^a_k)^{a=1,\dots,n}_{k\geq0}$, defined by 
\begin{equation}\label{preDSeq}
\frac{\partial \mathcal{L}(\lambda)}{\partial T^a_k}=\left[(\lambda^kR^{PDE}_a(\lambda))_+,\mathcal{L}(\lambda)\right]\,, \quad a=1,\dots,n,\,k\geq0 \,. 
\end{equation}
As in \cite{BDY2}, we call $R^{PDE}_a(\lambda)$ the {\it basic $\g$-resolvents} for the pre-Drinfeld--Sokolov hierarchy.

The $\tau$-structure \cite{BDY2} ({\it cf}. \cite{DLZ, DZ-norm}) $\Omega_{a,k;b,l}\in \mathcal{A}^q$, $a,b=1,\dots,n,\, k,l\geq0$, of pre-Drinfeld--Sokolov hierarchy can be defined via
\begin{equation}
	\sum_{k,l\geq 0}\frac{\Omega_{a,k;b,l}}{\lambda^{k+1}\mu^{l+1}}=\frac{\left(R^{PDE}_a(\lambda)|R^{PDE}_b(\mu)\right)}{\left(\lambda-\mu\right)^2}-\eta_{ab}\frac{m_a\lambda+m_b\mu}{\left(\lambda-\mu\right)^2}\,.
\end{equation} 
It has the following properties:
\begin{equation}
	\Omega_{a,k;b,l}=\Omega_{b,l;a,k},~~\frac{\partial\Omega_{a,k;b,l}}{\partial  T^c_m}=\frac{\partial\Omega_{c,m;a,k}}{\partial  T^b_l}=\frac{\partial\Omega_{b,l;c,m}}{\partial  T^a_k}\,,~~ a,b,c=1,\dots,n, \, k,l,m\geq 0\,.
\end{equation}
 For any solution $q(x,\mathbf{T})$ of \eqref{preDSeq}, there exists \cite{BDY2} a function $\tau(x,\mathbf{T})$ satisfying \begin{equation}\label{PDEtau}
	\frac{\partial^2\log\tau(x,\mathbf{T})}{\partial T^a_k\partial T^b_l}=\Omega_{a,k;b,l}(x,\mathbf{T}),\quad\frac{\partial \log\tau(x,\mathbf{T})}{\partial{T^1_0}}=-\frac{\partial\log\tau(x,\mathbf{T})}{\partial x}\,,
\end{equation}
by which the function $\tau(x,\mathbf{T})$ can be determied up to a factor as in \eqref{tauFactor}. In view of \eqref{PDEtau}, we can identify $x$ with $-T^1_0$, and use a shorter notation $\tau(\mathbf{T})$. This scalar function $\tau(\mathbf{T})$ is called {\it the $\tau$-function of $q(x,\mathbf{T})$ to the pre-Drinfeld--Sokolov hierarchy}. 

Let $\mathcal{N}(x)\in\mathfrak{n}$ be a function with its values in the nilpotent subalgebra. The gauge transformation $q(x)\mapsto \tilde{q}(x)$ is defined by \begin{equation}\label{eq:gauge11}
	\mathcal{L}(\lambda)~\mapsto~\tilde{\mathcal{L}}(\lambda)=e^{{\rm ad}_{\mathcal{N}(x)}}\mathcal{L}(\lambda)=\partial_x+\Lambda(\lambda)+\tilde{q}(x)\,.
\end{equation}
The following facts were proved   in \cite{DS, BDY2}:
\begin{itemize}
	\item Each gauge transformation is a symmetry of pre-Drinfeld--Sokolov hierarchy.
	\item The Drinfeld--Sokolov hierarchy can be obtained from \eqref{preDSeq} by a reduction via the gauge transformation. 
	\item The $\tau$-structure of the pre-Drinfeld--Sokolov hierarchy is invariant with respect to the gauge transformation.
\end{itemize}
Based on the above points, one can prove that the scalar function $\tau(\mathbf{T})$ (defined for the pre-Drinfeld--Sokolov hierarchy) is also a  $\tau$-function for the Drinfeld--Sokolov hierarchy \cite{BDY2}. The following lemma gives an algorithm for computing the Taylor coefficients  of $\log\tau(\mathbf{T})$ at $\mathbf{T}=\mathbf{0}$.
\begin{lem}[\cite{BDY2}]\label{TaylorTauPDE}
	Let $q(x,\mathbf{T})$ be an arbitrary solution of the pre-Drinfeld--Sokolov hierarchy, and let $\tau(\mathbf{T})$ be the $\tau$-function of this solution.  The following identity holds true for $m\geq 2$, $a_1,\dots,a_m=1,\dots,n$:
			\begin{align}
				\sum_{k_1,\dots,k_m\geq 0}\frac{\frac{\partial ^m\log\tau}{\partial T^{a_1}_{k_1}\dots\partial T^{a_m}_{k_m}}(\mathbf{0})}{\lambda_1^{k_1+1}\cdots\lambda_m^{k_m+1}}=-\frac{1}{2mh^{\vee}}\sum_{s\in S_m}\frac{B\big(R^{PDE}_{a_{s_1}}(\lambda_{s_1},\mathbf{0}),\dots,R^{PDE}_{a_{s_m}}(\lambda_{s_m},\mathbf{0})\big)}{\prod_{i=1}^{m}\left(\lambda_{s_i}-\lambda_{s_{i+1}}\right)}\nonumber\\
				-\delta_{m,2}\eta_{a_1a_2}\frac{m_{a_1}\lambda_1+m_{a_2}\lambda_2}{\left(\lambda_1-\lambda_2\right)^2}\,.
			\end{align}
Here, $s_{m+1}=s_1$, $R^{PDE}_a(\lambda,\mathbf{T})$ are the basic $\g$-resolvents in \eqref{PDEresolvent}, and $B$ is the $m$-linear form on $\g$ given in \eqref{mLinearFormB}. 
\end{lem}

\subsection{Proof of Theorem \ref{thm:3}}
In this subsection we prove Theorem \ref{thm:3}, which
establishes the relationship between $\tau$-functions for the  ODE system~\eqref{ODE intro} and those for the Drinfeld--Sokolov hierarchy.

\begin{proof}[Proof of Theorem \ref{thm:3}]
Let  $W(\lambda,\mathbf{T})$ be a solution to the  ODE system~\eqref{ODE intro}, and $\tau(\mathbf{T})$ the $\tau$-function of the solution $W(\lambda,\mathbf{T})$. Since \begin{equation}\label{R_1U}
	R_1(\lambda,\mathbf{T})=e^{{\rm ad}_{U{(\lambda,\mathbf{T})}}}\Lambda(\lambda)=\Lambda(\lambda)+\text{lower degree terms with respect to deg}\,,
\end{equation} we have
\begin{equation}
 R_1(\lambda,\mathbf{T})_+-\Lambda(\lambda)\in\mathfrak{b}[[\mathbf{T}]]\,.
\end{equation}
Here $U(\lambda,\mathbf{T})$ is given in Lemma \ref{FuLem}.
Let $q(x,\mathbf{T})$  be the unique solution to the pre-Drinfeld--Sokolov hierarchy \eqref{preDSeq} specified by the following initial data
 \begin{equation}\label{InData}
	q(x,\mathbf{T})|_{\mathbf{T}=\mathbf{0}}=R_1(\lambda;T^1_0=-x,0,0,\dots)_+-\Lambda(\lambda)\in \mathfrak{b}[[x]]\,,
\end{equation}
and ${R}^{\rm PDE}_a(\lambda,x,\mathbf{T})$, $a=1,\dots,n$, be the corresponding basic $\g$-resolvents of the pre-Drinfeld--Sokolov hierarchy.
From Lemma~\ref{LoopRa} we have
 \begin{equation}\label{dxR1}
	\frac{d R_1(\lambda,\mathbf{T})}{d T^1_0}=\left[ R_1(\lambda,\mathbf{T})_+,R_1(\lambda,\mathbf{T})\right]\,.
\end{equation}
With the help of this formula,  we can verify that \begin{equation}\label{LRad}
\left[\mathcal{L}|_{\mathbf{T}=\mathbf{0}},R_1(\lambda;T^1_0=-x,\mathbf{0})\right]=0\,,
\end{equation}
where we recall that the operator $\mathcal{L}$ is given in \eqref{LaxOP}.
Using \eqref{R_1U} and \eqref{LRad} one has \begin{equation}\label{LLambda0}
\big[e^{-{\rm ad}_{U(\lambda;-x,\mathbf{0})}}\mathcal{L}|_{\mathbf{T}=\mathbf{0}},\Lambda(\lambda)\big]=0\,.
\end{equation}
 Let us decompose $e^{-{\rm ad}_{U(\lambda;-x,\mathbf{0})}}\mathcal{L}|_{\mathbf{T}=\mathbf{0}}$ with respect to the principal gradation as follows:
\begin{equation}\label{ODEUHL}
e^{-{\rm ad}_{U(\lambda;-x,\mathbf{0})}}\mathcal{L}|_{\mathbf{T}=\mathbf{0}}=\partial_x+\Lambda(\lambda)+\sum_{k\geq 0}\mathcal{H}_k(\lambda,x),\quad \mathcal{H}_k(\lambda,x)\in \left(\mathfrak{g}((\lambda^{-1}))\right)^{-k}\,.
\end{equation}
Here we have used  the  fact that $\deg U(\lambda;-x,\mathbf{0})<0$. By \eqref{LLambda0} we have
 \begin{equation}
\mathcal{H}_k(\lambda,x)\in{\left(\rm Ker\,ad_{\Lambda(\lambda)}\right)}^{-k}\,,\quad k\geq 0\,.
\end{equation}
From Lemma \ref{FuLem} we know that
\begin{equation}\label{ODEUx}
U(\lambda;-x,\mathbf{0})=\sum_{k\geq 1}U^{[-k]}(\lambda;-x,\mathbf{0}),\quad U^{[-k]}(\lambda;-x,\mathbf{0})\in\left({\rm Im\,ad_{\Lambda(\lambda)}}\right)^{-k}\,.
\end{equation} 
 By using \eqref{ODEUHL}--\eqref{ODEUx} and the uniqueness of solution of   \eqref{UHformPDE}--\eqref{PDEflemmeq}, we have
 \begin{equation}
\mathcal{U}(\lambda,x)=U(\lambda;-x,\mathbf{0})\,.
\end{equation}
 Together with the definition \eqref{PDEresolvent} of the basic $\mathfrak{g}$-resovents of the pre-Drinfeld--Sokolov hierarchy, we have
\begin{equation}\label{RodeRpde}
R^{PDE}_a(\lambda,x)=R_a(\lambda;-x,\mathbf{0})\,, \quad a=1,\dots,n\,.
\end{equation}
Let $\tau^{PDE}(\mathbf{T})$ denote the $\tau$-function of the solution $q(x,\mathbf{T})$. The $m$-th order Taylor coeffients, $m\geq 2$,  of $\log\tau^{PDE}({\mathbf{T}})$ at $\mathbf{T}=\mathbf{0}$ are given in Lemma \ref{TaylorTauPDE}. And these coefficients are equal to those of $\log\tau(\mathbf{T})$, thanks to equaility \eqref{RodeRpde} and Corollary \ref{TaylorTauODE}.
We conclude that the function $\tau(\mathbf{T})$ equals to $\tau^{PDE}(\mathbf{T})$, up to a factor of the form  \eqref{tauFactor}. 

Conversely, let $q(x,\mathbf{T})$ be an arbitary solution to the pre-Drinfeld--Sokolov hierarchy \eqref{preDSeq} and $\tau^{PDE}(\mathbf{T})$ be its $\tau$-function. 
Let $W(\lambda,\mathbf{T})$ be a solution to the  ODE system~\eqref{ODE intro}, specified by the condition 
\begin{equation}
	W(\lambda,\mathbf{T})|_{\mathbf{T}=\mathbf{0}}=R^{PDE}_1(\lambda;x=0,{\mathbf{T}=\mathbf{0}})\,,
\end{equation} and $\tau(\mathbf{T})$ be the $\tau$-function of this solution.
For this initial data we have 
\begin{equation}
e^{-{\rm ad}_{\mathcal{U}(\lambda;x=0,\mathbf{T}=\mathbf{0})}}W(\lambda,\mathbf{0})=\Lambda(\lambda)\,.
\end{equation}
 The unequeness of solutions to \eqref{UHform-u}--\eqref{Flem} tells us
  \begin{equation}
	U(\lambda,\mathbf{0})=\mathcal{U}(\lambda;x=0,\mathbf{0})\,.
\end{equation}
Therefore,
\begin{equation}
	R_a(\lambda,\mathbf{0})=R^{PDE}_a(\lambda;x=0,\mathbf{0})\,,\quad a=1,\dots,n\,.
\end{equation}
Similar to the previous arguments, we obtain that $\tau(\mathbf{T})$ equals to $\tau^{PDE}(\mathbf{T})$, up to a factor like \eqref{tauFactor}.
\end{proof}

According to Theorem \ref{thm:3}, the functions $r_a$, $a=1,\dots,n$, defined by
 \begin{equation}
	r_a=\frac{\partial^2\log\tau(\mathbf{T})}{\partial T^{a}_0\partial T^1_0}=\omega_{a,0;1,0}\,,
\end{equation} 
satisfy 
\begin{equation}\label{DSnormalcoor}
	\frac{\partial r_a}{\partial T^b_k}=\partial_{T^1_0}\Omega_{a,0;b,k}\,,
\end{equation}
which is the Drinfeld--Sokolov hierarchy written in terms of normal coordinates \cite{BDY2,DZ-norm}.
In the next section we will interpret this by means of concrete examples.

\section{Examples}\label{sec:2}
 In this section, the  ODE systems~\eqref{ODE intro} for the $A_1$ and $A_2$ cases 
are explicitly computed, which are connected to the KdV hierarchy and the Boussinesq hierarchy respectively.

\subsection{ The infinite commuting ODE system of $A_1$-type}
The normalized Cartan--Killing form on~$\g=\mathfrak{sl}_2(\mathbb{C})$ is given by 
\begin{equation}
	\left(A|B\right)={\rm Tr}(AB),\quad A,B\in\mathfrak{sl}_2(\mathbb{C}).
\end{equation}
We take the usual basis of  $\mathfrak{sl}_2(\mathbb{C})$, which are homogeneous with respect to the principal gradation ({\it cf}. \eqref{Prin_degree}):
\begin{equation}
	e_1=\begin{pmatrix}
		0 & 1 \\
		0 & 0
	\end{pmatrix}\in \g^1,\quad e_2=\begin{pmatrix}
		1 & 0 \\
		0 & -1
	\end{pmatrix}\in \g^0, \quad e_3=\begin{pmatrix}
		0 & 0 \\
		1 & 0
	\end{pmatrix}\in\g^{-1}\,.
\end{equation}
In this case, the cyclic element $\Lambda(\lambda)$ of  $ \mathfrak{sl}_2(\mathbb{C})((\lambda^{-1}))$ is given by 
\begin{equation}
	\Lambda(\lambda)=
        \begin{pmatrix}
        	0 & 1 \\
		\lambda & 0 
        \end{pmatrix}\,,
\end{equation}
and the matrix-valued  formal Laurent series $W(\lambda)$ is given by
\begin{equation}\label{wkdv}
	W(\lambda)=
        \begin{pmatrix}
        	a_0 & 1 \\
		\lambda+c_1 & -a_0  
        \end{pmatrix}
	+\sum_{i\geq 1}
        \begin{pmatrix}
        	a_i & b_i \\
		c_{i+1} & -a_i  
        \end{pmatrix}
\lambda^{-i}\,. 
\end{equation}
Compared to \eqref{WDef},  $a_{k}=u^2_{k-1},b_{k+1}=u^1_{k},c_{k+1}=u^3_{k-1}$, $k\geq0$.

 The $r$-matrix \eqref{rmatrixDef}  reads (here an irrelevant term proportional to $I\otimes I$ is  added)    
\begin{equation}\label{StandardRA1}
	r(\lambda)=\frac{P }{\lambda}\,, \quad P(x\otimes y)=y\otimes x\,,\quad x,y\in \mathbb{C}^2\,. 
\end{equation}
The Poisson bracket \eqref{r-matixPoisson} reads 
	\begin{align}
		\left\{a(\lambda),b(\mu)\right\}=-\frac{b(\lambda)-b(\mu)}{\lambda-\mu}\,,~~ &\left\{a(\lambda),c(\mu)\right\}=\frac{c(\lambda)-c(\mu)}{\lambda-\mu}+1\,,\\
		\left\{b(\lambda),c(\mu)\right\}=-\frac{ 2a(\lambda)- 2a(\mu) }{ \lambda-\mu }\,, ~~  &\left\{a(\lambda),a(\mu)\right\}=\left\{b(\lambda),b(\mu)\right\}=\left\{c(\lambda),c(\mu)\right\}=0\,,
	\end{align}
with
 \begin{equation}
	a(\lambda)=a_0+\sum_{k\geq1}a_k\lambda^{-k}\,,\quad b(\lambda)=\sum_{k\geq1}b_k\lambda^{-k}\,,\quad c(\lambda)=\sum_{k\geq1}c_k\lambda^{-k+1}\,.
\end{equation}

	It follows from  \eqref{Sec3WhR} that the generating series of the hamiltonians  $h(\lambda)$  and basic the $\g$-resolvent $R(\lambda)$ satisfy 
	\begin{equation}
		W(\lambda)=\frac{h(\lambda)R(\lambda)}{2\lambda}\,.
	\end{equation} Thanks to $(R(\lambda)|R(\lambda))=2\lambda$ ({\it cf.}~Lemma \ref{lemRR}), we have closed forms for $h(\lambda)$ and $R(\lambda)$:
	\begin{equation}\label{WRHformula}
		h(\lambda)=\sqrt{2\lambda {\rm Tr}(W(\lambda)^2)}\,,\quad R(\lambda)=\sqrt{\frac{2\lambda}{{\rm Tr}\left(W(\lambda)^2\right)}}W(\lambda)\,.
	\end{equation}
	
	The first few hamiltonians are given by
	\begin{align*}
		h_{-1}&=a_{0}^{2}+b_{1}+c_{1},\\
		h_0&=-\frac{1}{4} a_{0}^{4}-\frac{1}{2} a_{0}^{2} b_{1}-\frac{1}{2} a_{0}^{2} c_{1}+2 a_{0} a_{1}-\frac{1}{4} b_{1}^{2}+\frac{1}{2} c_{1} b_{1}-\frac{1}{4} c_{1}^{2}+c_{2}+b_{2}\,,\\
		h_{1}&=-\frac{1}{2} b_{2} b_{1}-\frac{1}{8} c_{1}^{2} b_{1}-a_{1} a_{0}^{3}+\frac{1}{2} c_{1} b_{2}-\frac{1}{8} c_{1} b_{1}^{2}-\frac{1}{2} c_{1} c_{2}+\frac{1}{2} c_{2} b_{1}+\frac{3}{8} b_{1}^{2} a_{0}^{2}+\frac{3}{8} b_{1} a_{0}^{4}+\frac{3}{8} a_{0}^{2} c_{1}^{2}
		+\frac{3}{8} c_{1} a_{0}^{4}\\ & -\frac{1}{2} a_{0}^{2} b_{2}-\frac{1}{2} a_{0}^{2} c_{2}+2 a_{0} a_{2}-a_{0} b_{1} a_{1}+\frac{1}{4} b_{1} a_{0}^{2} c_{1}-a_{0} c_{1} a_{1}+\frac{1}{8} a_{0}^{6}+\frac{1}{8} c_{1}^{3}+a_{1}^{2}+\frac{1}{8} b_{1}^{3}+b_{3}+c_{3}\,,
	\end{align*}
	and the first few  flows are given by
	\begin{align}
		\frac{d a_0}{dT_0}&=\frac{a_{0}^{2}}{2}-\frac{b_{1}}{2}+\frac{c_{1}}{2}\,,\quad 		 \frac{d b_1}{d T_0}=2 a_{0} b_{1}-2 a_{1}\,,\quad \frac{d c_1}{d T_0}=-a_{0}^{3}-a_{0} b_{1}-a_{0} c_{1}+2 a_{1}\,, \label{DS_ODEabc0}\\
		\frac{d a_0}{dT_1}&=\frac{1}{4} c_{1} b_{1}-\frac{1}{2} b_{2}+\frac{1}{8} b_{1}^{2}+\frac{1}{2} c_{2}-\frac{3}{4} a_{0}^{2} c_{1}-\frac{3}{8} a_{0}^{4}-\frac{1}{4} a_{0}^{2} b_{1}+a_{0} a_{1}-\frac{3}{8} c_{1}^{2}\,,\\
		\frac{d b_1}{d T_1}&=-a_{0}^{3} b_{1}+a_{0}^{2} a_{1}-a_{0} b_{1}^{2}-a_{0} b_{1} c_{1}+2 a_{0} b_{2}+a_{1} b_{1}+a_{1} c_{1}-2 a_{2}\,,\\
		\frac{d c_1}{d T_1}&=-3 a_{0}^{2} a_{1}+\frac{3}{4} a_{0} b_{1}^{2}+\frac{3}{2} a_{0}^{3} b_{1}+\frac{3}{4} a_{0} c_{1}^{2}+\frac{3}{2} c_{1} a_{0}^{3}-a_{0} b_{2}-a_{0} c_{2}+2 a_{2}-a_{1} b_{1}+\frac{1}{2} a_{0} b_{1} c_{1}-a_{1} c_{1}+\frac{3}{4} a_{0}^{5}\,.\label{DS_ODEabc4}
	\end{align}

	Denote the basic $\mathfrak{sl}_2(\mathbb{C})$-resolvent  by\begin{equation}
		R(\lambda)=\left(\begin{array}{cc}
			\tilde{a}(\lambda) & \tilde{b}(\lambda) \\
			\tilde{c}(\lambda) & -\tilde{a}(\lambda)
		\end{array}\right)\,. 
	\end{equation}
Explicitly, one has 
\begin{align*}
		\tilde{a}(\lambda)=&a_0+(a_1-\frac{1}{2}a_0b_1-\frac{1}{2}a_0c_1-\frac{1}{2}a_0^3)\lambda^{-1}+\dots\,,\\
		\tilde{b}(\lambda)=&1+(\frac{b_1}{2}-\frac{a_0^2}{2}-\frac{c_1}{2})\lambda^{-1}+\dots\,,\\
		\tilde{c}(\lambda)=&\lambda+(\frac{c_1}{2}-\frac{a_0^2}{2}-\frac{b_1}{2})\\ & +(\frac{3}{4} a_{0}^{2} b_{1}+\frac{1}{4} a_{0}^{2} c_{1}-\frac{1}{4} c_{1} b_{1}-a_{0} a_{1}-\frac{1}{2} b_{2}+\frac{1}{2} c_{2}+\frac{3}{8} a_{0}^{4}-\frac{1}{8} c_{1}^{2}+\frac{3}{8} b_{1}^{2})\lambda^{-1}+\dots \,.
	\end{align*}
			
			Using \eqref{ODE_tauStucture}, we compute  the first few polynomials of the $\tau$-structure as follows:
			\begin{align*}
				\omega_{0,0}=&-\frac{a_{0}^{2}}{2}+\frac{b_{1}}{2}-\frac{c_{1}}{2}\,,\label{OmegA100}\\
				\omega_{0,1}=&-\frac{1}{4} c_{1} b_{1}+\frac{1}{2} b_{2}-\frac{1}{8} b_{1}^{2}-\frac{1}{2} c_{2}+\frac{3}{4} a_{0}^{2} c_{1}+\frac{3}{8} a_{0}^{4}+\frac{1}{4} a_{0}^{2} b_{1}-a_{0} a_{1}+\frac{3}{8} c_{1}^{2}\,,\\
				\omega_{1,1}=&-\frac{1}{8} b_{1}^{3}+\frac{1}{2} b_{3}-\frac{1}{2} c_{3}-\frac{3}{8} a_{0}^{6}-\frac{3}{8} c_{1}^{3}-\frac{3}{2} a_{1}^{2}+2 a_{0} b_{1} a_{1}-\frac{3}{4} b_{1} a_{0}^{2} c_{1}+2 a_{0} c_{1} a_{1}
				+\frac{1}{8} c_{1}^{2} b_{1}+ 2 a_{1} a_{0}^{3}\\ &  -\frac{1}{2} c_{1} b_{2}+\frac{3}{8} c_{1} b_{1}^{2}+c_{1} c_{2}-\frac{1}{2} c_{2} b_{1}-\frac{5}{8} b_{1}^{2} a_{0}^{2}-\frac{7}{8} b_{1} a_{0}^{4}-\frac{9}{8} a_{0}^{2} c_{1}^{2}-\frac{9}{8} c_{1} a_{0}^{4}+a_{0}^{2} c_{2}-a_{0} a_{2}\,.
			\end{align*}

		Let us illustrate the connection between the  infinite commuting ODE system of $A_1$-type and the KdV hierarchy by matching their $\tau$-structures. Set $u=\omega_{0,0}$. 
	From the  ODE system of $A_1$-type \eqref{ODE intro} ({\it cf.}~\eqref{DS_ODEabc0}--\eqref{DS_ODEabc4}), we get
			\begin{equation}\label{A1omega01}
				\omega_{0,1}=\frac{3}{2}u^2+\frac{1}{4}\frac{d^2 u}{d T_0^2}\,,\quad 	\omega_{1,1}=3u^3+\frac{3}{8}\left(\frac{d u}{d{T_0}}\right)^2+\frac{3}{2}u\frac{d^2 u}{d T_0^2}+\frac{1}{16}\frac{d^4 u}{d T_0^4}\,.
			\end{equation}
		Note that these expressions coincide with the $\tau$-structure of the KdV hierarchy, which is guaranteed by Theorem \ref{thm:3}.
				Given a solution $W(\lambda,\mathbf{T})$ to \eqref{ODE intro} of $A_1$-type, we could verify that $u$ satisfies the KdV equation:
				\begin{equation}\label{A1KdV}
					\frac{\partial u}{\partial T_1}=3u\frac{\partial u}{\partial T_0}+\frac{1}{4}\frac{\partial^3 u}{\partial T_0^3}\,. 
				\end{equation}
	 Similarly,  expressions for other $\omega_{k,l}$, $k,l\geq 0$, will lead to higher-order KdV equations.

\begin{rmk}\label{rmk:diff}
 The ODE version of the KdV hierarchy (the $A_1$ case) was originally set up by Dubrovin~\cite{Du2}. 
	The 
        series $W(\lambda)$ \eqref{wkdv} is connected to the one used in~\cite{Du2} by a gauge transformation
        					\begin{equation}
						W(\lambda) \mapsto G\,W(\lambda)\,G^{-1}\,,
						\end{equation}
where \begin{equation}
						G=\begin{pmatrix}
							1 & 0 \\
							a_0 & 1
						\end{pmatrix}\,.
					\end{equation}
For other types of simple Lie algebra, the r\^ole of the gauge transformation in the hamiltonian formalism will be studied separately. 
				\end{rmk}

\subsection{ The infinite commuting ODE system of $A_2$-type}
In this case $\mathfrak{g}=\mathfrak{sl}_3(\mathbb{C})$, and the normalized Killing form is 
\begin{equation}
	\left(A|B\right)={\rm Tr}(AB),\quad A,B\in\mathfrak{g}.
\end{equation}

Let $e_1, \dots, e_8$ be a basis of $\mathfrak{sl}_3(\mathbb{C})$ given by
 \begin{equation}
\begin{aligned}
	&e_1=E_{11}-E_{33},\quad e_2=E_{11}-E_{22},\quad e_3=E_{12},\quad e_4=E_{13},\quad\\
	& e_5=E_{21},\quad e_6=E_{23},\quad e_7=E_{31},\quad e_8=E_{32}\,,
\end{aligned}
\end{equation}
where $E_{ij}$ is the canonical basis of a $3\times 3$ matrix. 
The cyclic element is
 \begin{equation}
\Lambda(\lambda)=\begin{pmatrix}
	0 & 1 & 0 \\
	0 & 0 & 1 \\
	\lambda & 0 & 0
\end{pmatrix}\,,
\end{equation}
and the elements $\Lambda_1(\lambda),\Lambda_2(\lambda)$ are given by
\begin{equation}\label{A2Lambda12}
\Lambda_1(\lambda)=\Lambda(\lambda)\,,\quad \Lambda_2=\Lambda(\lambda)^2\,.
\end{equation}
The principal degree on $\mathfrak{g}$ reads \begin{equation}
\deg E_{ij}=j-i,\quad\deg \lambda=3\,.
\end{equation}

The matrix-valued  formal Laurent series $W(\lambda)$ is given by
\begin{equation}
	W(\lambda)=\Lambda(\lambda)+\sum_{i=1}^{8}u_i(\lambda)e_i=\Lambda(\lambda)+\text{lower degree terms with respect to deg}\,.
\end{equation}
Instead of using upper index for the indeterminates as in \eqref{WDef}, now we use lower index in this example. Precisely, 
\begin{equation}
u_i(\lambda)=\sum_{k\geq -1}u_{i,k}\lambda^{-k-1},i=1,2,5,7,8,\quad u_i(\lambda)=\sum_{k\geq 0}u_{i,k}\lambda^{-k-1},i=3,4,6.
\end{equation}

The first few hamiltonians defined by Lemma \ref{FuLem} and \eqref{H_hLambda_j} are explicitly as follows:
 \begin{align}
	\begin{autobreak}
			h_{1,-1}=
				u_{1,-1}^{2}
				+u_{1,-1} u_{2,-1}
				+u_{2,-1}^{2}
				+u_{4,0}
				+u_{5,-1}
				+u_{8,-1}\,,
			\end{autobreak}\label{A2flowabc}\nonumber \\
		\begin{autobreak}
			h_{2,-1}=
			u_{1,-1}^{2} u_{2,-1}
			+u_{1,-1} u_{2,-1}^{2}
			+u_{5,-1} u_{1,-1}
			-u_{8,-1} u_{1,-1}
			+u_{2,-1} u_{4,0}
			-u_{8,-1} u_{2,-1}
			+u_{3,0}
			+u_{6,0}
			+u_{7,-1}\,,
		\end{autobreak}\nonumber \\
		\begin{autobreak}
		h_{1,0}=
		u_{4,1}
		+u_{5,0}
		-\frac{1}{3} u_{4,0} u_{2,-1}^{3}
		+\frac{1}{3} u_{8,-1} u_{1,-1}^{3}
		+\frac{1}{3} u_{8,-1} u_{2,-1}^{3}
		-\frac{1}{3} u_{6,0} u_{2,-1}^{2}
		-\frac{1}{3} u_{7,-1} u_{2,-1}^{2}
		-\frac{1}{3} u_{5,-1} u_{1,-1}^{3}
		-\frac{1}{3} u_{3,0} u_{1,-1}^{2}
		-\frac{1}{3} u_{3,0} u_{2,-1}^{2}
		+\frac{2}{3} u_{5,-1} u_{3,0}
		-\frac{1}{3} u_{5,-1}^{2} u_{1,-1}
		-\frac{1}{3} u_{5,-1} u_{6,0}
		-\frac{1}{3} u_{5,-1} u_{7,-1}
		+\frac{2}{3} u_{6,0} u_{8,-1}
		-\frac{1}{3} u_{6,0} u_{1,-1}^{2}
		-\frac{1}{3} u_{8,-1} u_{3,0}
		-\frac{1}{3} u_{8,-1} u_{7,-1}
		+\frac{1}{3} u_{8,-1}^{2} u_{1,-1}
		+\frac{1}{3} u_{8,-1}^{2} u_{2,-1}
		-\frac{1}{3} u_{4,0}^{2} u_{2,-1}
		-\frac{1}{3} u_{1,-1}^{4} u_{2,-1}
		+u_{1,0} u_{2,-1}
		+u_{2,0} u_{1,-1}
		+2 u_{2,0} u_{2,-1}
		-\frac{2}{3} u_{1,-1}^{3} u_{2,-1}^{2}
		-\frac{2}{3} u_{1,-1}^{2} u_{2,-1}^{3}
		-\frac{1}{3} u_{7,-1} u_{1,-1}^{2}
		+2 u_{1,0} u_{1,-1}
		+\frac{2}{3} u_{4,0} u_{7,-1}
		-\frac{1}{3} u_{3,0} u_{4,0}
		-\frac{1}{3} u_{6,0} u_{4,0}
		-\frac{1}{3} u_{1,-1} u_{2,-1}^{4}
		+\frac{1}{3} u_{4,0} u_{8,-1} u_{1,-1}
		-\frac{2}{3} u_{4,0} u_{1,-1}^{2} u_{2,-1}
		-\frac{2}{3} u_{4,0} u_{1,-1} u_{2,-1}^{2}
		-\frac{2}{3} u_{5,-1} u_{1,-1}^{2} u_{2,-1}
		-\frac{2}{3} u_{5,-1} u_{1,-1} u_{2,-1}^{2}
		+\frac{1}{3} u_{1,-1}^{2} u_{2,-1} u_{8,-1}
		+\frac{1}{3} u_{1,-1} u_{2,-1}^{2} u_{8,-1}
		-\frac{1}{3} u_{4,0} u_{5,-1} u_{1,-1}
		-\frac{1}{3} u_{4,0} u_{5,-1} u_{2,-1}
		+\frac{1}{3} u_{8,-1} u_{5,-1} u_{2,-1}
		+u_{8,0}
		-\frac{1}{3} u_{6,0} u_{2,-1} u_{1,-1}
		-\frac{1}{3} u_{7,-1} u_{1,-1} u_{2,-1}
		-\frac{1}{3} u_{3,0} u_{1,-1} u_{2,-1}\,,
	\end{autobreak}\nonumber \\
\begin{autobreak}
 h_{2,0}=
	u_{3,1}+u_{6,1}
	-\frac{1}{9} u_{5,-1} u_{1,-1}^{4}
	-\frac{4}{9} u_{5,-1}^{2} u_{1,-1}^{2}
	-\frac{1}{9} u_{5,-1}^{2} u_{2,-1}^{2}
	-\frac{4}{9} u_{8,-1}^{2} u_{1,-1}^{2}
	-\frac{4}{9}u_{8,-1}^{2}u_{2,-1}^{2}
	+u_{1,0}u_{2,-1}^{2}
	+u_{2,0} u_{1,-1}^{2}
	-\frac{4}{9} u_{4,0}^{2} u_{2,-1}^{2}
	-\frac{1}{9} u_{4,0} u_{5,-1}^{2}
	+u_{4,0} u_{2,0}
	-\frac{1}{9} u_{4,0}^{2} u_{8,-1}
	-\frac{1}{9} u_{4,0} u_{8,-1}^{2}
	-\frac{1}{9} u_{8,-1} u_{5,-1}^{2}
	-u_{8,-1} u_{1,0}
	-\frac{1}{9} u_{4,0} u_{1,-1}^{4}
	-\frac{1}{9} u_{4,0}^{2} u_{1,-1}^{2}
	+\frac{1}{3} u_{3,0} u_{6,0}
	+\frac{1}{3} u_{3,0} u_{7,-1}
	+\frac{1}{3} u_{6,0} u_{7,-1}
	-\frac{1}{9} u_{1,-1}^{5} u_{2,-1}
	-\frac{1}{9} u_{2,-1}^{5} u_{1,-1}
	-\frac{5}{9} u_{1,-1}^{2} u_{2,-1}^{4}
	-\frac{5}{9} u_{1,-1}^{4} u_{2,-1}^{2}
	-\frac{25}{27} u_{1,-1}^{3} u_{2,-1}^{3}
	+u_{5,0} u_{1,-1}
	-\frac{1}{9} u_{4,0}^{2} u_{5,-1}
	-\frac{1}{9} u_{5,-1} u_{8,-1}^{2}+u_{5,-1} u_{1,0}
	+u_{4,1} u_{2,-1}
	-u_{8,0} u_{1,-1}
	-\frac{1}{9} u_{5,-1} u_{2,-1}^{4}
	-\frac{1}{9} u_{4,0} u_{2,-1}^{4}
	-\frac{1}{9} u_{8,-1} u_{1,-1}^{4}
	-\frac{1}{9} u_{8,-1} u_{2,-1}^{4}
	+u_{7,0}
	-\frac{1}{3} u_{6,0} u_{8,-1} u_{2,-1}
	+\frac{2}{3} u_{8,-1} u_{3,0} u_{2,-1}
	+\frac{2}{3} u_{7,-1} u_{8,-1} u_{2,-1}
	-\frac{8}{9} u_{4,0} u_{5,-1} u_{1,-1} u_{2,-1}
	+\frac{4}{9} u_{4,0} u_{8,-1} u_{1,-1} u_{2,-1}
	+\frac{4}{9} u_{5,-1} u_{8,-1} u_{1,-1} u_{2,-1}
	-\frac{1}{27} u_{1,-1}^{6}
	-\frac{1}{27} u_{2,-1}^{6}
	-\frac{1}{27} u_{8,-1}^{3}
	-\frac{1}{3} u_{3,0}^{2}
	-\frac{1}{27} u_{5,-1}^{3}
	-\frac{1}{3} u_{6,0}^{2}
	-\frac{1}{27} u_{4,0}^{3}
	-\frac{1}{3} u_{7,-1}^{2}
	-u_{8,0} u_{2,-1}
	-u_{8,-1} u_{2,0}
	+\frac{1}{3} u_{3,0} u_{5,-1} u_{1,-1}
	-\frac{2}{3} u_{4,0} u_{6,0} u_{2,-1}
	+\frac{1}{3} u_{4,0} u_{7,-1} u_{2,-1}
	-\frac{1}{3} u_{6,0} u_{8,-1} u_{1,-1}
	-u_{5,-1} u_{1,-1}^{2} u_{2,-1}^{2}
	-\frac{2}{9} u_{4,0} u_{1,-1}^{3} u_{2,-1}
	-u_{4,0} u_{1,-1}^{2} u_{2,-1}^{2}
	-\frac{2}{3} u_{3,0} u_{2,-1}^{2} u_{1,-1}
	-\frac{2}{9} u_{4,0} u_{5,-1} u_{1,-1}^{2}
	-\frac{2}{9} u_{4,0} u_{5,-1} u_{2,-1}^{2}
	-\frac{2}{9} u_{4,0} u_{8,-1} u_{1,-1}^{2}
	+\frac{4}{9} u_{4,0} u_{8,-1} u_{2,-1}^{2}
	+u_{8,-1} u_{1,-1}^{2} u_{2,-1}^{2}
	-\frac{2}{3} u_{6,0} u_{1,-1} u_{2,-1}^{2}
	-\frac{2}{3} u_{5,-1} u_{6,0} u_{1,-1}
	+\frac{2}{3} u_{8,-1} u_{3,0} u_{1,-1}
	-\frac{2}{9} u_{8,-1} u_{2,-1}^{2} u_{5,-1}
	-\frac{2}{3} u_{7,-1} u_{5,-1} u_{1,-1}
	+\frac{2}{3} u_{7,-1} u_{8,-1} u_{1,-1}
	-\frac{2}{9} u_{5,-1} u_{2,-1}^{3} u_{1,-1}
	-\frac{8}{9} u_{4,0} u_{2,-1}^{3} u_{1,-1}
	+\frac{4}{9} u_{8,-1} u_{1,-1}^{3} u_{2,-1}
	+\frac{4}{9} u_{8,-1} u_{2,-1}^{3} u_{1,-1}
	-\frac{2}{3} u_{6,0} u_{1,-1}^{2} u_{2,-1}
	-\frac{2}{3} u_{7,-1} u_{1,-1}^{2} u_{2,-1}
	-\frac{2}{3} u_{7,-1} u_{2,-1}^{2} u_{1,-1}
	-\frac{8}{9} u_{5,-1} u_{1,-1}^{3} u_{2,-1}
	-\frac{2}{3} u_{3,0} u_{1,-1}^{2} u_{2,-1}
	-\frac{1}{9} u_{5,-1}^{2} u_{1,-1} u_{2,-1}
	-\frac{7}{9} u_{8,-1}^{2} u_{1,-1} u_{2,-1}
	+2 u_{1,0} u_{1,-1} u_{2,-1}
	+2 u_{2,0} u_{1,-1} u_{2,-1}
	-\frac{2}{3} u_{4,0} u_{3,0} u_{2,-1}
	-\frac{1}{9} u_{4,0}^{2} u_{2,-1} u_{1,-1}
	+\frac{4}{9} u_{5,-1} u_{8,-1} u_{1,-1}^{2}
	+\frac{7}{9} u_{4,0} u_{8,-1} u_{5,-1}\,.
\end{autobreak}\nonumber
\end{align}

It follows from \eqref{DefBasicResolvent}, \eqref{Sec3WhR} and \eqref{A2Lambda12} that the generating series of hamiltonians $h_1(\lambda)$, $h_2(\lambda)$ defined in \eqref{hamiltonianSeries} satisfy the following relation
\begin{align}
	{\rm Tr}(W(\lambda)^2)=\frac{2 h_1(\lambda)h_2(\lambda)}{3\lambda}\,,\quad {\rm Tr}(W(\lambda)^3)=\frac{h_1(\lambda)^3+\lambda h_2(\lambda)^3}{9\lambda^2}\,.
\end{align}

Some first few polynomials of the $\tau$-structure are in the forms
\begin{align}
		\begin{autobreak}
			\omega_{1,0;1,0}=
			-\frac{1}{3} u_{8,-1}
			-\frac{1}{3} u_{5,-1}
			-\frac{1}{3} u_{1,-1}^{2}
			-\frac{1}{3} u_{2,-1}^{2}
			-\frac{1}{3} u_{1,-1} u_{2,-1}
			+\frac{2}{3} u_{4,0}\,,
		\end{autobreak}\nonumber\\
		\begin{autobreak}
			\omega_{1,0;2,0}=
			\frac{1}{3} u_{3,0}
			+\frac{1}{3} u_{6,0}
			+\frac{1}{3} u_{2,-1} u_{4,0}
			-\frac{2}{3} u_{1,-1}^{2} u_{2,-1}
			-\frac{2}{3} u_{1,-1} u_{2,-1}^{2}
			+\frac{2}{3} u_{8,-1} u_{1,-1}
			+\frac{2}{3} u_{8,-1} u_{2,-1}
			-\frac{2}{3} u_{5,-1} u_{1,-1}
			-\frac{2}{3} u_{7,-1}\,,
		\end{autobreak}\nonumber\\\begin{autobreak}
		\omega_{2,0;2,0}=
		-\frac{2}{9} u_{8,-1}^{2}
		+u_{3,0} u_{1,-1}
		-u_{2,0}
		-\frac{4}{9} u_{4,0} u_{1,-1} u_{2,-1}
		-\frac{4}{9} u_{5,-1} u_{1,-1} u_{2,-1}
		-\frac{4}{9} u_{1,-1} u_{2,-1} u_{8,-1}
		-\frac{2}{9} u_{4,0}^{2}
		-\frac{2}{9} u_{5,-1}^{2}
		-\frac{2}{9} u_{2,-1}^{4}
		-\frac{2}{9} u_{1,-1}^{4}
		-u_{1,-1} u_{6,0}
		-\frac{4}{9} u_{5,-1} u_{2,-1}^{2}
		-\frac{4}{9} u_{1,-1}^{3} u_{2,-1}
		-\frac{4}{9} u_{1,-1} u_{2,-1}^{3}
		-\frac{4}{9} u_{4,0} u_{1,-1}^{2}
		+\frac{5}{9} u_{4,0} u_{2,-1}^{2}
		+\frac{5}{9} u_{4,0} u_{8,-1}
		-\frac{4}{9} u_{8,-1} u_{1,-1}^{2}
		-\frac{4}{9} u_{8,-1} u_{2,-1}^{2}
		-\frac{4}{9} u_{8,-1} u_{5,-1}
		-\frac{4}{9} u_{5,-1} u_{1,-1}^{2}
		+u_{3,0} u_{2,-1}
		-\frac{2}{3} u_{1,-1}^{2} u_{2,-1}^{2}
		+\frac{5}{9} u_{4,0} u_{5,-1}\,.
	\end{autobreak}\nonumber
\end{align}

Set $r_1=\omega_{1,0;1,0}$, $r_2=\omega_{1,0;2,0}$. Then, one has
\begin{equation}\label{A2omega20}
\omega_{2,0;2,0}=-\frac{1}{3}	\frac{d^2 r_{1}}{d(T^1_0)^2}-2r_1^2\,.
\end{equation}
Using \eqref{ODE intro} together with the above expressions, one obtains the  Boussinesq equation
\begin{align}
	&\frac{\partial r_1}{\partial T^2_0}=\frac{\partial r_2}{\partial T^1_0}\,,\label{normalA2r1}\\
	&\frac{\partial r_2}{\partial T^2_0}=-\frac{1}{3}	\frac{\partial^3 r_{1}}{\partial(T^1_0)^3}-4r_1\frac{\partial r_1}{\partial T^1_0}\,.\label{normalA2r2}
\end{align}

\section{Conclusion}
\label{sec:concl}
As a natural continuation of a recent work set up by Dubrovin~\cite{Du2} in the study of the KdV hierarchy, an infinite family of (hamiltonian) pairwise commuting ODE system (with the hamiltonians being in involution) associated to any simple Lie algebra $\g$ was provided ({\it cf.}~\eqref{ODE intro}) in this paper. The $\tau$-functions for the ODE system were also defined, and  were shown to 
be those for the Drinfeld--Sokolov hierarchy of~$\g$-type (and vice versa).

Among many possible perspectives of the present work we can think of, we would like to mention the followings that are of importance, and that the ODE formalism could have an advantage in the studies of integrable systems:
\begin{itemize}
\item[1.] The hamiltonian structure of the ODE system  could be quantized using the canonical quantization. This  might also lead to a quantization procedure for the Drinfeld--Sokolov hierarchy and related integrable structures such as the $\tau$-functions. 
\item[2.] The construction of the theta-functions as special $\tau$-functions could be more straightforward using the ODE formalism. This was  already  shown by Dubrovin for the KdV case \cite{Du2, Du3}, but could now be applied to other  Drinfeld--Sokolov hierarchies.
  \item[3.] We could generalize the so-called Dubrovin equations \cite{Du0} in the study of finite-gap solutions of integrable PDEs to the infinite case, as for the KdV case a finite ODE system (using a truncated $W(\lambda)$) is closely related to the Dubrovin equations of finite-gap solutions  \cite{Du2}. 
\item[4.]
It would be interesting to construct the ODE systems and their $\tau$-functions with the Poisson bracket given by the $r$-matrix of other type (trigonometric or elliptic, cf.~\cite{BD-r-matrix, Sklyanin, FT-book}).
\end{itemize}
We hope to study the above questions in subsequent papers.

\section*{Acknowledgments}
D.~Yang is supported by NSFC (No.~12371254), and by the CAS Project for Young Scientists in Basic Research (No.~YSBR-032). C.~Zhang is supported by NSFC (No.~12171306). 
\appendix

\medskip

\medskip

\noindent Di Yang

\noindent {\rm   School of Mathematical Sciences, University of Science and Technology of China, 
 Hefei 230026, China}

\noindent  \url{diyang@ustc.edu.cn}

~~~~~

\noindent Cheng Zhang

\noindent {\rm Department of Mathematics, Shanghai University, Shanghai, 200444, China}

\noindent {\rm Newtouch Center for Mathematics of Shanghai University, Shanghai 200444, China}

\noindent \url{ch.zhang.maths@gmail.com}


~~~~~

\noindent Zejun Zhou

\noindent  {\rm School of Mathematical Sciences, University of Science and Technology of China, 
 Hefei 230026, China}

\noindent \url{zzj24601@mail.ustc.edu.cn}


\begin{thebibliography}{99}
\bibitem{BD-r-matrix}A. Belavin, V. Drinfeld (1982). Solutions of the classical Yang–Baxter equation for simple Lie algebras. {\em Functional Analysis and Its Applications},  16(3):159--180.


\bibitem{BDY}
M. Bertola, B. Dubrovin, D. Yang (2016). Correlation functions of the KdV hierarchy and applications to intersection numbers over $\overline{{\cal M}}_{g, n}$. {\em Physica D: Nonlinear Phenomena}, 327:30--57.

\bibitem{BDY2}
M. Bertola, B. Dubrovin, D. Yang (2021). Simple Lie algebras, Drinfeld--Sokolov hierarchies, and multipoint correlation functions. {\em Moscow Mathematical Journal}, 21(2):233--270.

\bibitem{Dickey}
L.A. Dickey (2003). Soliton Equations and Hamiltonian Systems, 2nd edition, World Scientific.

\bibitem{DS}
V. G. Drinfeld, V. V. Sokolov (1985). Lie algebras and equations of Korteweg--de Vries type. {\em Journal of Soviet mathematics},  30(1985):1975--2036.

\bibitem{Du0} B. Dubrovin (1975). Periodic problems for the Korteweg--de Vries equation in the class of finite band potentials. {\em Functional analysis and its applications}, 9(3):215--223.


 \bibitem{Du1}
B. Dubrovin (1996). Geometry of 2D topological field theories, in: M. Francaviglia,  S. Greco, (eds) {\em Integrable Systems and Quantum Groups}. {\em Lecture Notes in Mathematics}, vol 1620:120--348. Springer.

\bibitem{Du2} 
B. Dubrovin (2019). Approximating tau-functions by
  theta-functions. {\em Communications in Number Theory and Physics}, 13(1):203--223.
  
\bibitem{Du3}
B. Dubrovin (2020). Algebraic spectral curves over $\mathbb{Q}$ and their tau-functions, in: R. Donagi, T. Shaska, (eds) {\em  Integrable Systems and Algebraic Geometry}, vol 2:41--91. Cambridge University Press.


  
 \bibitem{DLZ}
B. Dubrovin, S.-Q. Liu, Y. Zhang (2008). Frobenius manifolds and central invariants for the Drinfeld--Sokolov bihamiltonian structures. {\em Advances in Mathematics}, 219(3):780--837.
 \bibitem{DZ-norm}
		Dubrovin, B., Zhang, Y. Normal forms of hierarchies of integrable PDEs, Frobenius manifolds and Gromov-Witten invariants. arXiv:math/0108160.

\bibitem{FSZ}
C. Faber, S. Shadrin, D. Zvonkine (2010). Tautological relations and the $r$-spin Witten conjecture. {\em Annales Scientifiques de l'Ecole Normale Sup\'erieure}, 43(4):621--658. 

\bibitem{FT-book} 
L. D. Faddeev, L. A. Takhtajan (2007). {\em Hamiltonian Methods in the Theory of Solitons}, Springer Science \& Business Media. 

\bibitem{FJR}
H. Fan, T. Jarvis, Y. Ruan (2013). {\em The Witten equation, mirror symmetry, and quantum singularity theory. Annals of Mathematics}, 178(1):1--106.

\bibitem{Kac1978}
V. G. Kac (1978). Infinite-dimensional algebras, Dedekind's $\eta$-function, classical M\"obius function and the very strange formula. {\em Advances in Mathematics}, 30(2):85--136.

\bibitem{Kac}
 V. G. Kac (1994). {\em Infinite-dimensional Lie Algebras},  Cambridge university press.

\bibitem{Kontsevich}
M. Kontsevich (1992). Intersection theory on the moduli space of curves and the matrix Airy function.  {\em Communications in Mathematical Physics}, 147(2):1--23.

\bibitem{Kostant}
B. Kostant (1959). The principal three-dimensional subgroup and the Betti numbers of a complex simple Lie group. {\em American Journal of Mathematics}, 81(4):973--1032.


\bibitem{LRZ}
S.-Q. Liu, Y. Ruan, Y. Zhang (2015). BCFG Drinfeld–Sokolov hierarchies and FJRW-theory. {\em Inventiones Mathematicae}, 201(2):711--772.

\bibitem{Sklyanin}  E. K. Sklyanin (1979).  On complete integrability of the Landau-Lifshitz equation. No. LOMI-79-3. 

 \bibitem{W1}
 E. Witten (1991). Two-dimensional gravity and intersection theory on moduli space. {\em Surveys in differential geometry}, 1(1):243--310.

\bibitem{W2}
E. Witten (1993). Algebraic geometry associated with matrix models of two-dimensional gravity, in:
 L. Goldberg, A. Phillips, (eds), {\em Topological Methods in Modern Mathematics}:235--269, Publish or Perish. 


\end{thebibliography}
\end{document}